\documentclass[12pt, fleqn]{article}

\usepackage[utf8]{inputenc}
\usepackage[T1]{fontenc}
\usepackage{url}
\usepackage{amsfonts, amsthm, amsmath, amssymb}
\usepackage{stix}
\usepackage{a4wide}
\usepackage{hyperref}
\usepackage{authblk}
\usepackage{xcolor}
\usepackage[ruled, linesnumbered, noend]{algorithm2e}

\hypersetup{
    colorlinks=true,
    urlcolor=blue,
    citecolor=blue,
    linkcolor=blue,
}

\title{Highway: Efficient Consensus with Flexible Finality}
\author[1]{Daniel Kane}
\author[2]{Andreas Fackler}
\author[3]{Adam Gągol}
\author[4]{Damian Straszak}

\affil[1]{Computer Science and Engineering Department, UC San Diego}
\affil[2]{CasperLabs AG}
\affil[3,4]{Cardinal Cryptography}

\newtheorem{theorem}{Theorem}
\newtheorem{lemma}{Lemma}

\newtheorem{definition}{Definition}

\newcommand{\ww}{\mathbb{w}}
\newcommand{\eendorse}{\mathrm{ENDORSE}}
\newcommand{\vote}{\mathrm{vote}}

\newcommand{\ghost}{\mathrm{GHOST}}
\newcommand{\summit}{\mathrm{SUMMIT}}

\newcommand{\opinion}{\mathrm{opinion}}
\newcommand{\final}{\mathrm{FINAL}}
\newcommand{\leader}{\mathrm{LEADER}}
\newcommand{\gst}{\mathrm{GST}}

\newcommand{\nmin}{n_{\mathrm{min}}}
\newcommand{\nmax}{n_{\mathrm{max}}}
\newcommand{\equivs}{\mathit{Equivocators}}
\newcommand{\cntsucc}{\mathrm{cnt}_{\mathrm{succ}}}

\newcommand{\nxt}{\mathrm{next}}
\newcommand{\defeq}{\stackrel{\textup{def}}{=}}

\newcommand{\N}{\mathbb{N}}

\newcommand{\cB}{\mathcal{B}}
\newcommand{\cV}{\mathcal{V}}
\newcommand{\cL}{\mathcal{L}}
\newcommand{\cS}{\mathcal{S}}

\newcommand{\inparen}[1]{\left(#1\right)}
\newcommand{\inbraces}[1]{\left\{#1\right\}}

\newcommand{\desc}[2][0.95]
{
\vspace{5mm}
  {\centering
  \fbox{ 
    \small
    \begin{minipage}[c]{#1\linewidth}
      #2
    \end{minipage}
  }}
  \vspace{5mm}
}

\begin{document}

\maketitle

\begin{abstract}
There has been recently a lot of progress in designing efficient partially synchronous BFT consensus protocols that are meant to serve as core consensus engines for Proof of Stake blockchain systems.
While the state-of-the-art solutions attain virtually optimal performance under this theoretical model, there is still room for improvement, as several practical aspects of such systems are not captured by this model.
Most notably, during regular execution, due to financial incentives in such systems, one expects an overwhelming fraction of nodes to honestly follow the protocol rules and only few of them to be faulty, most likely due to temporary network issues.
Intuitively, the fact that almost all nodes behave honestly should result in stronger confidence in blocks finalized in such periods, however it is not the case under the classical model, where finality is binary.

We propose Highway, a new consensus protocol that is safe and live in the classical partially synchronous BFT model, while at the same time offering practical improvements over existing solutions.
Specifically, block finality in Highway is not binary but is expressed by fraction of nodes that would need to break the protocol rules in order for a block to be reverted.
During periods of honest participation finality of blocks might reach well beyond $1/3$ (as what would be the maximum for classical protocols), up to even $1$ (complete certainty).
Having finality defined this way, Highway offers flexibility with respect to the configuration of security thresholds among nodes running the protocol, allowing nodes with lower thresholds to reach finality faster than the ones requiring higher levels of confidence.

\end{abstract}

\section{Introduction}

Since the introduction of Bitcoin~\cite{nakamoto2008bitcoin} and the concept of a decentralized, tamperproof database -- a blockchain -- a number of different paradigms have been developed to design such databases.
Recently, the idea of building such systems based on PoS (Proof of Stake) has gained significant popularity.
While in the original PoW (Proof of Work, as used in Bitcoin) mechanism that is used for incentivizing participation and securing the system, the voting power of a participant is proportional to the amount of computational power possessed, in PoS the voting power is proportional to the amount of tokens (digital currency specific to this system).
A popular choice in such systems is then to periodically delegate a fixed size committee of participants which then is responsible for running the consensus on which blocks to add to the blockchain.
This way of building a blockchain has two substantial advantages over vanilla PoW systems such as Bitcoin: 1) it allows to run one of the classical permissioned consensus protocols that have been developed over the last 4 decades, 2) it allows to not only reward nodes for participation but also penalize misbehavior, by slashing security deposits of the offending committee members.

There has been recently tremendous progress in the design of permissioned consensus protocols that can be used as core engines in such PoS blockchains~\cite{AMNRY19,BKM18,BG17,CS20,GLSS19,GAGMPRSTT19,YMRGA19,zamfir2018casper}.
A vast majority of them are designed in the partially synchronous BFT model~\cite{DLS88} which asserts that communication between nodes becomes eventually synchronous and that no more than a given fraction of nodes, say $1/3$ (which is optimal in this model), are dishonest and may violate the protocol in an arbitrary way.
State-of-the-art protocols such as Hotstuff~\cite{YMRGA19}, Tendermint~\cite{BG17} and Streamlet~\cite{CS20} come close to optimality with respect to bandwith, latency of finalization and, also importantly, simplicity.
However, there are several practical properties of such blockchain systems that are not captured by this classical model, and consequently, significant room for improvement remains.
One such important aspect is that the partition of nodes into honest and Byzantine might not accurately reflect their true attitude.
In fact, according to the model, even ``honest'' nodes that have missed several protocol messages because of a DDoS attack or even a temporary network failure, are considered Byzantine.
In a situation where more than $1/3$ of nodes suffered (even for a few seconds) from such a network issue, protocols in the classical BFT model are not guaranteed to function properly.

On the other hand, besides these occasional offline periods, it is fair to assume that in a real-world system an overwhelming fraction, if not all, of the nodes honestly follow the protocol rules.
This is a consequence of the financial incentives for honest participation.
Indeed, it is in the best interest of committee members to make sure they actively participate in the consensus protocol, as they are paid a salary for honest work and are penalized for being offline or not contributing enough to the protocol progress.
In fact, because of penalties for protocol offences, it is highly unlikely that an adversary tries an attack which is not guaranteed to succeed, as otherwise it risks significant losses.
Therefore, with the only exception of large-scale, coordinated attacks that are intended to bring down the whole system, one should always expect almost all nodes to behave honestly.

Motivated by this realization there have been several works that design protocols which are safe in the classical sense while at the same time trying to offer better guarantees in ``typical'' scenarios.
In this paper we propose a new protocol -- Highway -- that contributes to this line of work.
The security of Highway is still formalized on grounds of the partially synchronous BFT model, thus in particular it achieves safety and liveness in the most demanding setting when $1/3$ of all nodes are Byzantine.
However, on top of that, Highway offers the following two features that make it particularly attractive in real-world deployments.
First of all, in periods of honest participation of a large fraction of nodes, it allows to reach finality of blocks with ``confidence'' much higher than the typical threshold of $1/3$.
To give an example, if a block reaches finality confidence of $0.8$ (which is possible in Highway) then at least $80\%$ of the nodes would need to violate the protocol in order to revert the block from the chain.
This stands in contrast with the classical notion of finalization that is binary: either a block is finalized (this means finality confidence of $1/3$) or it is not.
The second practical improvement in Highway is that it achieves flexibility akin to the notion defined in~\cite{MNR19}.
The nodes participating in Highway might be configured with different security trade-offs between the allowed number of Byzantine and crashing nodes (nodes that might go offline but are otherwise honest) in the protocol.
Flexibility then means that despite these differences in configuration, all the nodes run a single version of the protocol and perform the same actions, only the finality decisions they make depend on the chosen parameters.
A practical consequence is that nodes with lower security thresholds might reach finality much faster than nodes with higher thresholds, but as long as both these nodes' assumptions are satisfied they finalize the same blocks and stay in agreement.
    
Technically, Highway can be categorized as a DAG-based protocol~\cite{baird2016hashgraph,GLSS19,moser1999byzantine,zamfir2018casper}, in which nodes jointly maintain a common history of protocol messages, forming a directed acyclic graph representing the causality order.
In its design, Highway derives from the CBC-Casper approach~\cite{zamfir2018casper} and significantly improves upon it by the use of a new finality mechanism, message creation schedule and spam prevention mechanism.
We believe that the conceptual simplicity of DAG-based protocols along with the desirable practical features of the Highway protocol make it a solid choice for a consensus engine in a Proof of Stake-based blockchain.

\section{Our Results}

\subsection{Model}

We consider a system with a fixed set $\mathcal{V}$ of $n$ validators, each of them equipped with a public key that is known to other nodes.
This model matches the scenario of ``permissioned blockchain'', but the protocol can be applied to semi-permissionless scenario as well by rotating the set of validators.
Our model makes the following assumptions:
\begin{itemize}
    \item\emph{(Reliable point-to-point communication)}
    We assume that channels do not drop messages and all messages in the protocol are authenticated by digital signature of the sender.
    \item\emph{(Partially synchronous network)} There exists a publicly known bound $\Delta$ and an unknown Global Stabilization Time (GST) so that after GST, whenever a validator sends a message, it reaches the recipient within time $\Delta$. Additionally, we assume that validators have bounded clock drift\footnote{Note that bounded clock drift can be achieved in any partially synchronous network by means of Byzantine clock synchronization such as \cite{DLS88}}. 
    Such version of partial synchrony is known as a \emph{known $\Delta$ flavour}, for the discussion on the version without publicly known $\Delta$, see Subsection \ref{sec:dynamicrounds}.
    \item\emph{(Byzantine faults)} We assume that $f$ out of $n$ validators are under total control of an adversary, and hence can arbitrarily deviate from the protocol. We do not make a global assumption on the relation between $f$ and $n$, as safety and liveness require different bounds, and the latter have an interaction with number of crashing nodes as well.
    \item\emph{(Crashing faults)} We assume that $c$ out of $n$ nodes can become permanently unresponsive at some point in the protocol execution.
\end{itemize}

\subsection{Consensus in the Context of Blockchain}

In a typical blockchain system, validators are tasked with performing an iterative consensus on an ever-growing chain of transactions that they receive from the external environment. 
In the process, they enclose transactions into blocks forming a blockchain, in which each block refers to its predecessor by hash.
The first one, the \emph{genesis block}, is part of the protocol definition. 

In Highway, as the set of validators is either constant or subject only to very controlled changes (between different eras, see Subsection \ref{sec:eras}), specific validators are directly appointed to construct a block in a given time slot.
They do so by enclosing transactions from their local queue and hash of the block that they believe should be the predecessor.

As it may happen that a validator does not refer to the last constructed block (either intentionally, or due to a network failure), the set $\mathcal{B}$ of blocks is a tree, with a unique path leading from each block to the root: the genesis block $G$. 
The main goal of the consensus protocol in such a scenario is to choose a single branch from such a tree.
For a block $B$, we refer to all the blocks that are on the other branches as \emph{competing} with $B$, as if any of them would be chosen, $B$ could not.

\subsection{Practical Challenges}

{\bf Strong optimistic finality.} While since the initial definition of the partially synchronous model by Dwork, Lynch and Stockmeyer~\cite{DLS88} a vast body of research was created to optimize various parameters of protocols in this setting, most of it was written under the semi-formal assumption that the existence of more than $n/3$ dishonest nodes predates the existence of any provable guarantees for such protocols.
Such an assumption stems from the fact that, as proven in the original paper, it is not possible for a partially synchronous protocol to guarantee both liveness and finality if $n<3f+1$.

In spite of this negative result, it is however possible to provide additional finality guarantees in the scenario where for a prolonged period even dishonest nodes do not deviate from the protocol and actively work to finalize blocks.
Although it may not seem like an important observation from the perspective of classical security models, we stress that it carries significant practical consequences -- in most cases of blockchain deployments it may be assumed that most of the time basically all the nodes will actively collaborate to achieve consensus\footnote{As most systems reward and punish validators based on their behavior, offline nodes and small-scale attacks (ones that do not succeed in reverting a finalized block) are strongly disincentivised and hence not common.}.
It is hence possible to provide much stronger finality guarantees for blocks created during such periods, so to revert them much more than $n/3$ validators would have to collude.

{\bf Flexibility.}
Once the classical $n<3f+1$ bound is left aside, a natural tradeoff between finality and liveness occurs -- the stronger finality guarantee we require, the more validators need to honestly collaborate to finalize the block with such a guarantee.
The tradeoff could be resolved by all the validators agreeing on a common finality threshold that they intend to use, but such solution would have significant limitations.

In Highway, however, there is no need to agree upon a common threshold, so every validator is able to use a different one, or even several different thresholds.
Besides eliminating the need to make an additional consensus on this particular hyperparameter, one important implication of such feature is that it allows validators to play slightly different roles in the ecosystem -- for example some validators may deal mainly with finalizing relatively small transactions, in which case small latency is more important than very high security (and, as will become apparent after the protocol is presented, reaching higher thresholds usually takes more time), while others can prioritize safety over latency\footnote{In fact, in Highway the choice of specific threshold influences only local computations performed by a validator on the output of its communication with other validators. Hence, if validators would hand such communication logs to outside observers, observers could reinterpret the logs using different thresholds.} .

\subsection{Our Contribution}
We present Highway - a consensus protocol achieving strong optimistic finality that is flexible by allowing validators to use different confidence thresholds to convince themselves that a given block is ``finalized'' (both confidence threshold and finality will be properly defined in Subsection \ref{sec:finality}). 

Unless some validators actively deviate from the protocol, the finality of a block can only increase for a given validator, which intuitively corresponds to the ever-increasing number of confirmations for a block in PoW scenario.
However, unlike in PoW, in Highway the confidence levels for a given block can be directly interpreted as the number of validators that would need to misbehave in order to reverse such a block, what we formalize as the following theorem:

\begin{theorem}[\bf Finality]\label{thm:finality}
If an honest validator reaches finality with confidence threshold $t\geq f$ for a given valid block $B$, then no honest validator will ever reach finality with confidence threshold $t$ for a block competing with $B$.
\end{theorem}

Note that while due to the aforementioned impossibility result \cite{DLS88} it is not possible to prove liveness for confidence threshold $\frac{n}{3}$ and higher, in practice the vast majority of validators will not deviate from the protocol most of the times.
In such times, arbitrarily high confidence thresholds can be reached, which makes block constructed during such periods virtually impossible to revert.

Next, we provide a bound on the number of honest validators needed to guarantee that the protocol will continue finalizing blocks with given confidence threshold.
We note that it is in line with the classical $n\geq 3f+1$ bound with the added notion of ``crashing faults'', denoted by $c$, which disrupts the consensus process, but not as much as the Byzantine ones. 

\begin{theorem}[\bf Liveness]\label{thm:liveness}
For every confidence threshold $0 \leq t < \frac{n}{3}$, if $f\leq t$ and $c< \frac{n-3t}{2}$, then the chain of blocks finalized with confidence $t$ grows indefinitely for each honest validator.
\end{theorem}

\subsection{Related Work}

The line of work on partially synchronous protocols was initiated with the seminal work of Dwork, Lynch and Stockmeyer \cite{DLS88}, and gained popularity with the introduction of PBFT\cite{CL99} protocol and its numerous versions~\cite{BKM18,KADCW09,MNR19}. 
Classically, protocols in this model attain resilience against less than $n/3$ malicious nodes, due to a known bound stating that is it not possible to provide both safety and liveness with higher number of Byzantine faults \cite{DLS88}.
Some of the works however explore the concept of ``flexibility'' understood usually as providing the strongest $n/3$ security in the general partially synchronous model, and additional guarantees in case some additional conditions, such as network synchronicity or limited adversarial behavior, are met.

Gasper\cite{BHKPQRSWZ18}, the newly proposed candidate for Ethereum 2.0 consensus mechanism, analyzed the case of additional guarantees in case the network satisfies synchronicity assumptions. Later very similar considerations got a proper formal treatment with the introduction of the snap-and-chat family of protocols\cite{NNT20}.
The snap-and-chat protocols define two ``levels of finality'' (formalized in the paper as two ledgers, one being an extension of another). The first one, faster, relies on the classical partially synchronous assumptions and is guaranteed to be live and safe as long as less than $n/3$ nodes are faulty. 
The second level provides a stronger $n/2$ resilience against adversarial nodes, but is live only as long as the network is synchronous.
As in practical deployment assuming network synchronicity requires rather pessimistic assumptions about network latency, the second finality can be assumed to progress significantly slower. 
The provided construction is very modular and allows to use a wide variety of protocols to provide for the first and second finality levels, and the consistency between both levels is guaranteed.

Protocol that carries perhaps the most similarities with Highway when it comes to achieved security guarantees is Flexible BFT\cite{MNR19}.
It defines the new type of faulty node, the alive-but-corrupt node, that fully cooperates with honest nodes unless it is able to perform a successful attack on protocol safety.
Intuitively, it models the practical situation in which nodes without an explicit incentive do not cheat, as that would mean loosing rewards in the PoS system. 
Similarly as in Highway, the protocol is able to tolerate much more adversarial nodes if they do not aim to merely break the liveness, but are interested only in breaking safety - the bounds match the bounds in our paper. 
Flexible BFT also introduces, as the name suggests, certain flexibility for the nodes when it comes to choosing the parameters related to the finality - each of the nodes can have independent assumptions about number of faulty nodes of each kind, and it is guaranteed that two honest nodes with correct assumptions can't finalize competing blocks, and if all nodes have correct assumptions, the protocol will continue making progress.
The biggest conceptual difference is that in Flexible BFT, assumptions held by nodes explicitly influence their behavior in the protocol, while in Highway the assumed confidence threshold influences only the local computations performed on the DAG whose form is not influenced by specific choices of confidence thresholds made by validators. 
There are two main consequences of this difference. First, in Highway validators can update their confidence thresholds and they are able to recompute finality of all of the blocks without the need of communicating with other validators, while in case of Flexible BFT that would require rerunning the whole protocol. 
Second, perhaps even more importantly, Flexible BFT can stall for everyone in case some of the honest validators incorrectly assume the number of faulty parties.
In contrast, in Highway unit production never stalls, and the consensus may stall from the perspective of a given validator only if that specific validator made incorrect assumptions on the number of faulty parties.

To illustrate this difference, consider the scenario in which there is a big group of overly-conservative honest nodes incorrectly assuming that $90\%$ of the nodes are honest -- in such scenario in Flexible BFT even less conservative honest nodes will not be able to finalize blocks, while less conservative honest validators in Highway will not be influenced by such a choice of other validators, as it does not influence the communication between them in any way -- in fact, validators doesn't even have explicit means of checking confidence thresholds chosen by the others.

\section{Protocol}
We denote by $G$ the ``genesis block'' of the blockchain, which is considered part of the protocol definition.
Except the genesis block, every other block $B$ consists of a reference to its parent, denoted $\mathrm{prev}(B)$, and the content of the block  -- typically a list of transactions.
The parent reference in $B$ is realized by including the hash of $\mathrm{prev}(B)$ in $B$ and thus there cannot be any cycles in the block graph.
We also denote by $\mathrm{next}(B)$ the set of all blocks for which $B$ is the parent.
We recursively define the height $H$ of a block: the height $H(G)$ of the genesis block is $0$ and $H(B)=1+H(\mathrm{prev}(B))$ for any other block $B$.
We say that a block $B_1$ is a descendant of $B_2$ and write $B_2\leq B_1$ in case when one can reach $B_1$ from $B_2$ by following parent links (in particular $H(B_2)\leq H(B_1)$.

\subsection{Building a DAG}

In the Highway protocol, validators exchange messages in order to reach consensus on proposed blocks and hence validate one of possibly many branches of the produced blockchain.
As a way of capturing and spreading the different validators' knowledge about the already existing messages, it adopts the DAG framework~\cite{baird2016hashgraph,GLSS19,moser1999byzantine,zamfir2018casper}, in which every message broadcast by a validator refers a certain set of  messages sent by validators before.
We will refer to such messages broadcast during normal protocol operation as \emph{units}, and the included references as \emph{citations}. More formally, each unit consists of the following data
\begin{itemize}
    \item {\bf Sender.} The ID of the unit's sender (creator).
    \item {\bf Citations.} A list of hashes of other units that the creator wants to attest to. 
    \item {\bf Block.} In case a unit is produced by the validator appointed to produce a block at a given time, it is included in the unit. 
\end{itemize}

All units to be considered correct must also have a digital signature by its sender.
We denote the sender (creator) of a unit $u$ by $S(u)$.
As a given unit can only refer to previously constructed units, the citations contained in units can be seen as edges in a DAG (directed acyclic graph).

In the protocol it will be often important not only whether a given unit $u$ cites directly some other unit $u'$, but whether there is a provable casual dependence between them, i.e., whether it can be proven that during the creation of $u$ its creator was aware of $u'$.
Such a notion is easily captured by the existence of the chain of citations connecting $u$ and $u'$.
In the presence of such a chain, we say that the unit $u'$ is a  \emph{justification} of unit $u$ or, in other words, that $u'$ justifies $u$.
As the justification relation is transitive, we interpret it as a partial order on the set of units and denote the fact that $u'$ justifies $u$ by $u'\leq u$.
We denote the set of all units strictly smaller than $u$ in this order, i.e., its \emph{downset}, by $D(u)$, and the set $D(u)\cup\{u\}$ as $\bar{D}(u)$.
We also naturally extend the downset notation to sets of units, i.e., for a set $\cS$ of units we define
$$D(\cS) \defeq \bigcup_{u\in \cS} D(u)~~~~~~~\mbox{ and }~~~~~~~~\bar{D}(\cS) \defeq \bigcup_{u\in \cS} \bar{D}(u).$$
To formalize the notion of a protocol view in the DAG framework, we introduce the following definition:

\begin{definition} {\bf (Protocol State.)}
A finite set of units $\sigma$ is a \emph{protocol state} if it is closed under $D$, i.e., if $D(u)\subseteq \sigma$ for every $u\in\sigma$.
\end{definition}

When creating a new unit $u$, a validator $V$ is expected to always cite the last unit it created before $u$, and hence include all its previous units in $D(u)$.
As a consequence, the units created by honest validators always form a chain.
Note that malicious nodes can still deviate from this rule, and hence the following definition.

\begin{definition}{\bf (Equivocation.)}
A pair of units $(u,u')$ is an \emph{equivocation} if $S(u)=S(u')$ and $u$ and $u'$ are incomparable by the relation $\leq$. In such a case, the sender $S(u)$ is called an \emph{equivocator}.
\end{definition}

For a given set of units $\cS$ we denote the set of proven equivocators below it as $$E(\cS) \defeq \{V\in \mathcal{V} \mid \text{There exists } u,u'\in D(\cS) \text{ created by } V \text{ s.t. } u\not\leq u' \text{and } u'\not\leq u\}$$

We also denote the set of latest messages under a unit $u$ produced by honest (so far) validators as 
$$L(u) \defeq \{v\in D(u) \mid \cS(v)\notin E(\{u\}) \text{ and } v'>v \!\Rightarrow\! \cS(v')\neq \cS(v) \text{ for every } v'\in D(u)\}.$$

\subsection{Voting via the GHOST Rule}
In this section we introduce the GHOST (Greedy Heaviest Observed Sub-Tree) rule for fork selection in blockchains and explain how one can concretely implement it using an idea called virtual voting in the DAG.

\paragraph{The GHOST rule.}
An important task run by every blockchain client is that of fork selection.
Given a set of blocks $\cB$ that do not necessarily form a single chain, but are typically a tree of blocks, the goal is to pick a single ``tip'' (i.e., a leaf in this tree) to be considered as the head of the blockchain.
In systems based on Proof of Work, such as Bitcoin or Ethereum, the Longest Chain Rule is most commonly employed for that purpose, i.e., the leaf block of maximum depth is chosen as the head.

In our setting, we would like the fork selection rule to somehow express the common belief of the committee of validators on which chain is the ``main one''.
This however is not possible in the absence of additional information on what the validators ``think''.
Suppose therefore that besides the tree of blocks $\cB$ we are given opinions of validators, represented by a mapping 
$$\opinion: \cV \to \cB,$$
meaning that from the perspective of validator $V\in \cV$ the block $\opinion(V)$ should be the head of the blockchain. 
Having the block tree and opinions we are ready to define the GHOST rule.

\desc{
\begin{center}
    {\bf GHOST rule for block tree $\cB$ and $\opinion$ function} 
\end{center}
\begin{enumerate}
\item For each $B\in \cB$ compute $\mathrm{total}(B)$ to be the total number of validators $V\in \cV$ such that $\opinion(V) \geq B$ (i.e., such that $\opinion(V)$ is a descendant of $B$).
\item Set $B$ to the genesis block. Repeat the following steps while $B$ is not a leaf in $\cB$:
\begin{enumerate}
    \item Choose $B'\in \nxt(B)$ with largest $\mathrm{total}(B')$ (break ties by hash of $B'$).
    \item Set $B:=B'$.
\end{enumerate}
\item Output $B$.
\end{enumerate}
}

For brevity we write $\ghost(\cB, \opinion)$ to be the block resulting from applying the GHOST rule to the tree of blocks $\cB$ using the $\opinion:\cV \to \cB$.

\paragraph{Virtual Voting using the DAG.}

The GHOST rule is quite natural, intuitive and -- as we will soon demonstrate -- has a number of desirable properties.
However, for a node to apply it, it needs the tree of blocks and the ``most recent'' votes by each of the validators. 
The question becomes: how should the nodes maintain their local views on the block tree and what should they consider as the current opinions of other nodes? 
Note that even minor inconsistencies between two nodes on what the current opinions are might cause the GHOST rule to output different heads of the blockchain.
What is even worse is that such inconsistencies might be generated purposely by dishonest nodes, by sending different opinions to different validators, or sending opinions selectively.

In Highway, the consistency between local views of nodes is achieved with the help of the DAG.
Roughly speaking: each unit $u$ carries a virtual GHOST vote, depending only on $D(u)$ and possibly the block included in $u$ (if there is one). This vote is determined automatically from the virtual GHOST votes in the ``latest messages'' under $u$.
For brevity, below we denote by $L_V(u)$ the unique unit $v\in L(u)$ created by $V$ in case it exists, otherwise $L_V(u)=\bot$.
To formally define what a unit $u$ considers as $\vote(u)$, first we define an opinion function $\opinion_u: \cV \to \cB_u$ that is ``local to $u$''. Here $\cB_u$ consists of all blocks that appear in $\bar{D}(u)$, so it may happen that unit $u$ votes for a block that it carries.:

$$
    \opinion_u(V) \defeq \begin{cases}
    \vote\inparen{L_V(u)} &\mbox{if }  L_V(u)\neq \bot\\
    G &\mbox{otherwise}.
    \end{cases}
$$
Subsequently, we define $\vote(u)$ as
$$\vote(u) \defeq \ghost(\cB_u, \opinion_u).$$

While it does not affect correctness of the protocol, it is best for efficiency if honest validators always make sure that the block $B$ proposed in their unit $u$ satisfies $\vote(u)=B$.
This is achieved by choosing the parent of $B$ to be $\ghost(\cB_u \setminus \{B\}, \opinion_u)$ and can be also made a necessary condition for correctness of $u$.

One can interpret $\vote(u)$ as the block which the creator of $u$ considers as the head of the blockchain, at the moment when $u$ was created.
This does not quite mean that the creator of $u$ is certain that $\vote(u)$ will be ever finalized.
Instead, each validator maintains for every block a confidence parameter that indicates how likely it is that a given block will be ``reverted'', i.e., will not end up as part of the blockchain.
As we explain in the next section, this confidence parameter is proportional to how many validators would need to equivocate their units, in order to revert a given block.

\subsection{Finality Condition}\label{sec:finality}

Having defined the DAG and voting mechanism, we are ready to introduce the rules of finalizing blocks in the Highway protocol.
By taking advantage of the DAG framework, validators are able to compute finality of each block performing only local operations, namely searching for specific structures in their local copy of the DAG.

\begin{definition}
A $(q,k)$-summit for the block $B$, relative to protocol state $\sigma$ is a nested sequence of sets of units $(C_0, C_1, C_2, \dots, C_k)$ such that $C_0\supseteq C_1 \supseteq \ldots \supseteq C_k$ and:

\begin{itemize}
  \item {\bf unanimity.} $\vote(u) \geq B$ for all $u\in C_0$,
  \item {\bf honesty.} $E(\sigma) \cap S(C_0) = \emptyset$,
  \item {\bf convexity.} $u_0, u_2 \in C_i$ implies $u_1 \in C_i$, for all $u_0  \leq u_1 \leq u_2$ s.t. $S(u_1) = S(u_2) = S(u_3)$,
  \item {\bf density.} $\big|S\big(\bar{D}(u) \cap C'_i)\big)\big| \geq q$ for all $u \in C_{i + 1}$.
\end{itemize}
Where $C'_i = \big\{u\in C_i \mid \text{ there exists } u'\in C_{i+1} \text{ s.t. } S(u) = S(u') \big\}$.
\end{definition}

Intuitively, summits represent uninterrupted streaks of units produced by a big subset of validators (i.e., a quorum) that vote for the same block $B$ --- a structure in DAG after which it is very unlikely that the votes will change, hence summits will be used to finalize blocks.
In particular, as we will prove in the series of lemmas, checking the following condition will suffice to convince the validator that the given block $B$ will not be retracted, unless more than $t$ validators equivocate:

$$
   \final(B, \sigma, t) \defeq \begin{cases}
    1 &\mbox{if } \text{there exists $(q,k)$-summit for $B$ in $\sigma$ s.t. }(2q-n)(1-2^{-k})>t\\
    0 &\mbox{otherwise}.
    \end{cases}
$$

Now, based on the above definition we will refer to validator $V$ as having block $B$ \emph{finalized for confidence threshold} $t$ if $t$ is an integer value s.t. $\final(B,\sigma,t) = 1$, where $\sigma$ is $V$'s protocol state.
Next we prove the key technical lemma, showing that after a summit of certain height and quorum size occurs, new units will not vote against the vote of the summit.

\begin{lemma}\label{LemmaSummit}
Let $C = (C_0,\dots,C_k)$ be a $(q,k)$-summit for block $B$, relative to protocol state $\bar{D}(C_0)$, $\sigma\supseteq\bar{D}(C_0)$ be any correct protocol state, and let $u$ be a unit in $\sigma$ such that $D(u)\cap C_k\neq\emptyset$ and $\vote(u)\not\geq B$.
Then $f\geq (2q-n)(1-2^{-k})$.
\end{lemma}

\begin{proof}

Let us define 

$$A(u) = E(u)\cup \Big(S\big(\{v\in L(u) \mid \vote(v) \not\geq B\}\big) \cap E\big(C_0 \cup D(u)\big)\Big)$$

We will prove a slightly stronger thesis, namely that if the conditions of the Lemma are met, then $|A(u)| \geq (2q-n)(1-2^{-k})$.

We proceed with an induction on $k$.
Note that the $k = 0$ case follows trivially, as $(2q-n)(1-2^{-0}) = 0$, hence we assume that the statement holds for $k-1$.

Let us assume that there is a unit $u$ satisfying the assumptions of the Lemma, i.e., such that $D(u)\cap C_k\neq\emptyset$ and $\vote(u)\not\geq B$.
Let us take a minimal such $u$.

Let $u_k\in D(u)\cap C_k$ and $\mathcal{V}_{k-1} = S\big(\{v\in C'_{k-1}\mid v< u \}\big)$.
As $u$ is above some unit in $C_k$, we clearly have $| \mathcal{V}_{k-1}|\geq q$.
If all of the validators from $\mathcal{V}_{k-1}$ would keep voting for $B$ or for blocks above it, so would $u$, hence some of them must have equivocated or changed vote.
Let us additionally define $\mathcal{V}^{\mathrm{eq}}_{k-1} =\mathcal{V}_{k-1}\cap  E(u)$ and $\mathcal{V}^{\mathrm{change}}_{k-1} = \mathcal{V}_{k-1}\cap S\Big(\{v\in L(u)\mid \vote(v) \not\geq B\}\Big)$.
As $u$ is above at least $q-|\mathcal{V}^{\mathrm{eq}}_{k-1}|-|\mathcal{V}^{\mathrm{change}}_{k-1}|$ votes for $B$ and still votes against it, we have\footnote{Note that the total number of votes counted by $u$ is $n-|E(u)|$}:

\begin{equation*}
q-|\mathcal{V}^{\mathrm{eq}}_{k-1}|-|\mathcal{V}^{\mathrm{change}}_{k-1}| \leq \frac{n-|E(u)|}{2} \leq \frac{n-|\mathcal{V}^{\mathrm{eq}}_{k-1}|}{2}
\end{equation*}

\begin{equation}\label{EquationSummit1}
2q-n \leq |\mathcal{V}^{\mathrm{eq}}_{k-1}| +2|\mathcal{V}^{\mathrm{change}}_{k-1}|
\end{equation}

\noindent If $\mathcal{V}^{\mathrm{change}}_{k-1}$ is empty, we have the thesis already, so let us assume that $\mathcal{V}^{\mathrm{change}}_{k-1}\neq\emptyset$.
Let then $u'$ be a minimal unit such that $u'\leq u$, $\vote(u')\not\geq B$ and $D(u')\cap C_{k-1}\neq\emptyset$, which is bound to exist by noneptiness of $\mathcal{V}^{\mathrm{change}}_{k-1}$.
From the inductive assumption we get 

\begin{equation}\label{EquationSummit2}
    |A(u')| \geq (2q-n)(1-2^{-k+1})
\end{equation}

Next, connecting Inequalities \ref{EquationSummit1} and \ref{EquationSummit2}, we get

\begin{flalign*}
|A(u')\cup\mathcal{V}^{\mathrm{eq}}_{k-1} |+|\mathcal{V}^{\mathrm{change}}_{k-1}|\geq \\
\frac{|A(u')|}{2}+\frac{|\mathcal{V}^{\mathrm{eq}}_{k-1}|+2|\mathcal{V}^{\mathrm{change}}_{k-1}|}{2}\geq \\
\frac{(2q-n)(1-2^{-k+1})}{2}+\frac{2q-n}{2} = (2q-n)(1-2^{-k})
\end{flalign*}

\bigskip

Note that each $V\in \mathcal{V}^{\mathrm{change}}_{k-1}$ had to produce a unit in $C_k$, but also a unit witnessing voting for a block competing with $B$. 
If the latter unit would be above the unit in $C_k$, that would contradict the minimality of $u$, hence $V$ needs to belong to $E\big(C_0\cup D(u)) \big)$ and, in consequence, to $A(u)$.
As we also have $A(u')\subseteq A(u)$ and $\mathcal{V}^{\mathrm{eq}}_{k-1}\cap \mathcal{V}^{\mathrm{change}}_{k-1}=\emptyset$, the remaining thing to show is that the sets $A(u')$ and $\mathcal{V}^\mathrm{{change}}_{k-1}$ are disjoint.

Suppose that there exists $V\in A(u')\cap \mathcal{V}^{\mathrm{change}}_{k-1}$.
As $\mathcal{V}^{\mathrm{change}}_{k-1}\cap E(u)=\emptyset$, it has to be $V\in E\big(C_0\cup D(u')\big)$ and the last $V$'s unit that is strictly below $u'$ must vote for some block $B'$ competing with $B$.
Let $v$ be this unit voting for $B'$ and $v'\in C'_{k-1}$ - witness of $V$ being in $\mathcal{V}_{k-1}$.
As $V$ is not seen as an equivocator by $u'$, we have $v'\leq v$, hence we have $\vote(v)\not\geq B$, $D(v)\cap C_{k-1}\neq\emptyset$ and $v\leq u$ while $v<u'$, contradicting the minimality of $u'$.

\bigskip

\end{proof}

\begin{lemma}[Finality]\label{LemmaFinality}
Let $V_1$ and $V_2$ be honest validators with their respective protocol states $\sigma_1, \sigma_2$, and let $B_1, B_2$ be two competing blocks.
Then, if $\ \final(B_1,\sigma_1,t_1) = 1$ and $\final(B_2,\sigma_2,t_2) = 1$, it has to be $f> \min(t_1,t_2)$.
\end{lemma}

\begin{proof}

Let $\sigma = \sigma_1\cup\sigma_2 \cup \{u_{max}\}$, where $u_{max}$ is an additional unit created by one of the honest nodes which has full $\sigma_1\cup\sigma_2$ as its downset.
Note that $\sigma$ is an artificial state that might not necessarily come up in a real execution of the protocol, however it is still a correct state, thus Lemma~\ref{LemmaSummit} applies to it.
By definition of $\final$, we get that there is a $(q_1,k_1)$-summit for $B_1$ relative to $\sigma_1$ and a $(q_2,k_2)$-summit for $B_2$ relative to $\sigma_2$.
By definition we have: 

\begin{equation}\label{EquationFinality1}
(2q_1-n)(1-2^{-k_1})>t_1
\end{equation}

\vspace*{-0.9cm}

\begin{equation}\label{EquationFinality2}
(2q_2-n)(1-2^{-k_2})>t_2
\end{equation}

Note that even though $\sigma$ might not be reachable in protocol execution, Lemma \ref{LemmaSummit} holds for any correct state, as is $\sigma$.
Additionally, by the definition of the summit, any summit $(C_0,\dots,C_k)$ respective to some protocol state $\sigma'$ is also a summit respective to protocol state $\bar{D}(C_0)$.
As $u_{max}$ is above both of the summits, from Lemma \ref{LemmaSummit} we get that if $\vote(u_{max})\not\geq B_1$ then $f\geq (2q_1-n)(1-2^{-k_1})>t_1$, and if  $\vote(u_{max})\not\geq B_2$ then $f\geq (2q_2-n)(1-2^{-k_2})>t_2$.
As $B_1$ and $B_2$ are competing, $\vote(u_{max})$ can't be above both of them, hence we get $f>\min(t_1,t_2)$.

\end{proof}

Finally, we prove Theorem \ref{thm:finality} as a simple corollary:

\begin{proof}
Assume for contradiction that two validators have reached finality for competing blocks $B$ and $B'$ with confidence threshold $t\geq f$.
But then, by Lemma \ref{LemmaFinality}, we have that $f>\min(t,t)=t$, what finishes the proof.
\end{proof}

\subsection{Computability of the Finality Condition}
In this section we show that for a given state $\sigma$ and a quorum parameter $q$ a straightforward greedy algorithm can be used to find the highest possible $(q,\cdot)$-summit relative $\sigma$. 
We start with the pseudocode of this algorithm.

\desc{
\begin{center}
    {\bf Algorithm $\summit(\sigma, B, q)$ } 
\end{center}
\begin{enumerate}
\item Let $S_0\subseteq \cV$ be the set of all validators who have not equivocated in $\sigma$ and whose latest messages vote for $B$.
\item Let $C_0$ be the set of units created by $S_0$, built as follows: for each validator $V\in S_0$
\begin{enumerate}
    \item let $u$ be the latest unit created by $V$ in $\sigma$
    \item while $u$ votes for $B$:
    \begin{itemize}
        \item add $u$ to $C_0$
        \item replace $u$ by the direct parent of $u$ created by $V$ in $u$'s downset
    \end{itemize}
\end{enumerate}

\item For $l=1,2,3, \ldots$ repeat:
\begin{enumerate}
    \item Set $S_l:=S_{l-1}$
    \item Repeat:
    \begin{enumerate}
        \item Set $W:=S_l$
        \item Let $C_W:=\{u\in C_{l-1}: S(u) \in W\}$.
        \item For each $V\in W$:
        \begin{itemize}
            \item If there is no $u\in C_{l-1}$ by $V$ with $|D(u)\cap C_W|\geq q$, remove $V$ from $S_l$.
        \end{itemize}
        \item If $W=S_l$, break the loop.
    \end{enumerate}
    \item If $S_l=\emptyset$ return $\{C_0, C_1, \ldots, C_{l-1}\}$.
    \item Let $C_{l-1}':=\{u\in C_{l-1}: S(u) \in S_l\}$.
    \item Let $C_l:=\{u\in C_{l-1}: |D(u)\cap C_{l-1}'| \geq q, S(u) \in S_l\}$
\end{enumerate}
\end{enumerate}
}

The following lemma shows that the $\summit$ algorithm always outputs valid $(q,\cdot)$-summits and moreover that the returned summits are maximal.

\begin{lemma}[Properties of SUMMIT]
Let $(C_0, C_1, \ldots, C_r)$ be the  output of $\summit(\sigma, B, q)$ then
\begin{enumerate}
    \item {\bf (Correctness.)} If $r\geq 1$ then $(C_0, C_1, \ldots, C_r)$ is a $(q,r)$-summit for block $B$ relative to $\sigma$.
    \item {\bf (Maximality.)}  For every $(q,k)$-summit $(D_0, D_1, \ldots, D_k)$ for block $B$, relative to $\sigma$, we have $r\geq k$ and $D_i\subseteq C_i$ for each $i=0,1, \ldots, k$.
\end{enumerate}
\end{lemma}
\begin{proof}
\noindent {\bf Correctness.} It is evident by steps 1. and 2. of the algorithm that $C_0$ satisfies convexity and that all the $C_i$'s, for $i=0,2, \ldots, r$, as subsets of $C_0$ satisfy honesty and unanimity. Density for $C_i$ for $i>0$ is guaranteed in step 3.(e) of the algorithm, this way of constructing $C_i$ also implies convexity. 

What remains to prove is that all $C_i$'s are non-empty.
We proceed inductively: $C_0$ is non-empty because $r\geq 1$. Suppose now that $C_{i-1}\neq \emptyset$ for $1<i\leq r$, we show that $C_{i}\neq \emptyset$. Note that since $r\geq i$ means that $S_i\neq \emptyset$ in step 3.(c) of the algorithm. Moreover, the loop 3.(b) guarantees that after reaching 3.(c), each validator $V\in S_k$ has created at least one unit $u$ that satisfies $|D(u)\cap C'_{i-1}|\geq q$.
It follows that $C_i$ is non-empty.

\noindent {\bf Optimality.} We proceed by induction and show that $D_i \subseteq C_i$ and $S(D_i) \subseteq S(C_i)$ for each $i=0,1, \ldots, k$. Because of the convexity requirement for summits and the steps 1. and 2. of the algorithm, the base case $i=0$ holds.

Suppose now that $D_i \subseteq C_i$ and $S(D_i) \subseteq S(C_i)$ holds for $i=0,1,2, l-1$ for some $l\leq k$. We show that $D_l \subseteq C_l$ and $S(D_l) \subseteq S(C_l)$.
To this end, we first show the following invariant: the set $W$ that appears in the loop 3.(b) of the algorithm during iteration $l$ satisfies $S(D_l)\subseteq W$.
In  the first iteration we have $W=S_{l-1}$, and the invariant holds, since $S(D_l)\subseteq S(D_{l-1})$ and by induction $S(D_{l-1})\subseteq S(C_{l-1})=S_{l-1}$.
Now, whenever a validator $V\in W$ is removed from $W$ in step iii. of loop 3.(b), we have that no unit $u$ by $V$ in $D_{l-1}$ satisfies $|D(u) \cap C_{k-1}|\geq q$. Since $S(D_{l-1}) \subseteq W$, this means that also no $V$'s unit satisfies $|D(u) \cap D_{l-1}'|\geq q$ and in particular $V\notin S(D_l)$. 

In other words, whenever we remove a validator $V$ from the set $W$, this validator cannot belong to $S(D_l)$, thus still $S(D_l) \subseteq W \setminus \{V\}$.
Consequently, the invariant holds, which in turn implies that $S(D_l) \subseteq S_l =S(C_l)$.
The inclusion $D_l \subseteq C_l$ is a simple consequence of $S(D_l) \subseteq S(C_l)$ and the way $C_l$ is constructed in step 3.(e) of the algorithm.
\end{proof}

A straightforward consequence is that the finality criterion $\final(B, \sigma, t)$ is efficiently computable as follows

\desc{
\begin{center}
    {\bf Computing $\final(\sigma, B, t)$ } 
\end{center}
\begin{enumerate}
\item For $q=\lceil \frac{n+t}{2} \rceil,\lceil \frac{n+t}{2} \rceil+1, \ldots , n$ do:
\begin{enumerate}
    \item Let $k$ be the height of the summit output by $\summit(\sigma, B, q)$.
    \item If $(2q-n)(1-2^{-k})>t$ then {\bf return }$1$.
\end{enumerate}
\item {\bf return} $0$.
\end{enumerate}
}

\subsection{Guaranteeing Liveness}\label{sec:unit_creation}
So far we have discussed structural properties of the DAG produced during protocol execution and showed that certain combinatorial structures -- summits -- when found within the DAG can be used to conclude finality of a block.
What however is still missing is a proof that such combinatorial structures will actually appear in the DAG and that the adversary cannot indefinitely prevent any progress from happening.
Towards such a proof we need to first define a strategy for how and when validators should produce new units and when to include new blocks in them.

\noindent {\bf Unit Creation Schedule.}
Recall that we work in the partially synchronous model of communication in which a value $\Delta>0$ is known such that after an unknown moment in time called GST, each message from an honest validator takes less than $\Delta$ time to reach its recipient.
We make the standard assumption that validators have perfectly synchronized clocks (as each bounded difference in their local times can be instead treated as delay in message delivery).

We divide the protocol execution into rounds $r=0, 1, 2, \ldots$ of length $R:=3\Delta$. Each validator keeps track of rounds locally, based on its clock. 
We assume that there exists a leader schedule $\leader : \N \to \cV$ which assigns to a round index $r\in \N$ a round leader $\leader(r)\in \cV$.
Round leaders are responsible for creating new blocks, hence, as one can imagine, it is important to make sure that honest leaders appear in this schedule as often as possible.
In theory we only need to assume that honest validators appear in this schedule infinitely many times.
In practice one can, for instance, use a round-robin schedule or a pseudorandomly generated schedule.

We are now ready to introduce the strategy of how an honest validator should behave when it comes to creating new units

\desc{

\begin{center}
    {\bf Unit Creation and Reception Strategy by $V\in \cV$ in round $r\in \N$ } 
\end{center}
\tcc{Time is measured from the start of round $r$.}
\tcc{Whenever a new unit is created by $V$, it's above all maximal units in the local DAG, and it is added to the DAG right away.}
\begin{enumerate}
\item At time $t=0$: if $V=\leader(r)$
\begin{enumerate}
    \item move all units from the {\it buffer} to the local DAG.
    \item create and broadcast a new unit $u_L$ (referred to as the {\bf proposal unit}). Include in $u$ a new block $B$, choosing its parent so that $\vote(u_L)=B$.
\end{enumerate}
\item In the time slot $(0, R/3)$, if $V\neq \leader(r)$: if a new unit $u_L$ created by $\leader(r)$ is received, add it (along with its downset) to the local DAG right away. Immediately after adding $u_L$, create and broadcast a new unit (referred to as the {\bf confirmation unit} by $V$ in this round). All the remaining units received in this slot are placed in the {\it buffer}.
\item At time $R/3$ move units from the {\it buffer} to the local DAG.
\item In the time slot $(R/3, 2R/3)$: whenever a new unit is received, add it to the local DAG right away.
\item At time $t=2R/3$: create and broadcast a new unit (referred to as the {\bf witness unit} by $V$ in this round).
\item In the time slot $(2R/3, R)$: all the units received in this slot are placed in the {\it buffer}.
\end{enumerate}
}

Note that it might happen that a validator skips the creation of the confirmation unit (in case it did not receive a unit from the leader before time $R/3$).
The witness unit, on the other hand, is always created.
Thus in every round, the round leader creates $2$ units: the proposal unit, and the witness unit, and a non-leader creates either $1$ (the witness unit) or $2$ units (the confirmation and the witness unit).

\begin{lemma}\label{lemma:final_follow}
Let $t\geq f$ be an integer confidence parameter and assume that the number of crashing nodes is $c< \frac{n-3t}{2}$. If an honest validator reaches a state $\sigma$ such that $\final(\sigma, B, t)$ for some block $B$, then after some point in time, for every state $\sigma'$ reached by any honest validator it holds $\final(\sigma', B, t)$.
\end{lemma}

\begin{proof}
Denote by $H$ the set of all validators that are neither Byzantine nor crashing. As $n$ and $t$ are integer by definition, we have $c\leq \frac{n-3t-1}{2}$ and thus

\begin{equation}\label{eq:size_H}
 |H|=n-f-c\geq n-t-\inparen{\frac{n-3t-1}{2}} =  \frac{n+t+1}{2}.
\end{equation}

Note first that after $\gst$ each validator $V\in H$ eventually reaches a state that fully contains $\sigma$.
Let $r_0$ be the first round when this happens, from Lemma~\ref{LemmaSummit} it follows that each unit $u$ created by any of these validators from this point on satisfies $\vote(u) \geq B$.

Let us denote by $u^V_r$ the witness unit (i.e. the one created at time $2R/3$) created in round $r$ by the validator $V\in H$.
For $r\geq r_0+1$ we have the following properties satisfied by any $u\in \sigma$ created by $V$ with $u\geq u_{r}^V$
\begin{enumerate}
    \item $u$ contains in its downset all units $u^W_{r-1}$ for $W\in H$,
    \item $\vote(u) \geq B$.
\end{enumerate}
The former property follows from the fact that $r_0$ happened after $\gst$ and the latter is a consequence of Lemma~\ref{LemmaSummit}: indeed finality was reached because of the existence of a particular summit -- each unit $u$ built beyond this summit must then vote for a block  $\geq B$, thus $\vote(u)\geq B$.
It is now easy to see that the family of units $\{u_r^V\}$ can be used to build arbitrarily high summits.
More specifically, given a state $\sigma$ define the following sets:
$$C_i:= \inbraces{u \in \sigma: S(u)\in H, u\geq u_{r_0+i}^{S(u)}} ~~~~~~~\mbox{for }i\geq 0,$$
thus the set $C_i$ is the union of chains of units that begin with $u_{r_0+i}^V$ for each validator $V\in H$.
We claim that whenever 
\begin{equation}\label{eq:high_sigma}
    |S(C_k)|=|H|,
\end{equation} then $(C_0, C_1, \ldots, C_k)$ is a $(|H|, k)$-summit relative to $\sigma$.
This follows straight from the two properties listed above.
Moreover, since all the validators in $H$ are honest, each validator in $H$ will eventually reach a state $\sigma$ for which~\eqref{eq:high_sigma} holds.

It remains to note that, since by~\eqref{eq:size_H}, $|H|\geq  \frac{n+t+1}{2}$, for $k>\log(t+1)$ we have
$$(2|H|-n)(1-2^{-k}) \geq (t+1)(1-2^{-k})>t.$$

Consequently, $\final(B, \sigma', t)$ holds for each honest validator's state $\sigma'$ after $\log(n)$ rounds.

\end{proof}

Proof of Theorem \ref{thm:liveness}:

\begin{proof}
Note that by Lemma~\ref{lemma:final_follow} whenever even a single validator finalizes a block then all the others, eventually, will follow in finalizing exactly the same block.
Suppose thus, for the sake of contradiction that there exists a run of the protocol such that for some height $h\geq 0$, some block $B_h$ is finalized, but no block is ever finalized (by any honest validator) at height $h+1$.

Denote by $H$ the set of all validators that are neither Byzantine nor crashing. As in the proof of Lemma~\ref{lemma:final_follow} we have that $|H|\geq  \frac{n+t+1}{2}.$
Let $r_0$ be a round index such that:
\begin{enumerate}
    \item round $r_0-1$ happened after $\gst$,
    \item each validator $V\in H$ has already finalized $B_h$ before the beginning of round $r_0-1$, and consequently voted for a block above $B_h$ in all its units in the previous round,
    \item The leader $V_L:=\leader(r_0)$ is a member of $H$.
\end{enumerate}
The existence of such a round $r_0$ follows from Lemma~\ref{LemmaSummit} (each unit that is built upon the summit finalizing $B_h$ votes for a block $\geq B_h$) and the fact that validators in $H$ appear infinitely often in the schedule.
We show that the block $B$ proposed by the leader $V_L$ in round $r_0$ will get eventually finalized (thus reaching contradiction with the fact that height $h+1$ is never reached).

Towards this end, let us denote by $\sigma_0$ the state of $V_L$ at the start of round $r_0$, by $u_0$ the proposal unit created by $V_L$ at the start of round $r_0$, and by $B$ the block that $V_L$ proposes in $u_0$.
Since $V_L$ chooses as $B$'s parent the GHOST choice, we have that $\vote(u)=B$.
Moreover, $B_h \leq B$ because all latest units by validators in $H$ (which were created in round $(r_0-1)$) vote for blocks above $B_h$.

Consider now any validator $V\in H$ and let $u_V$ be the first of the units created by $V$ during round $r_0$ (i.e., proposal unit in case $V$ is the leader of that round, otherwise the confirmation unit of $V$, created before time $R/3$).
We claim that $$D(u_V) = \sigma_0 \cup \bar{D}(u_0).$$
This follows from the fact that $\gst$ has passed and from the unit creation and reception strategy.
Indeed: between time $2R/3$ in round $r_0-1$ when the previous unit $\bar{u}_V$ by $V$ was created, and the moment $<R/3$ in round $r_0$ when $V$ received $u_0$ ($R/3=\Delta$ hence $V$ is guaranteed to receive $u_0$ before time $R/3$), $V$ did not include in its DAG any other units than $u_0$ and its downset.
Moreover, $V_L$ received $\bar{u}_V$ before creating $u_0$ and hence $\bar{u}_V \in D(u_0)$, and thus the claim follows.

From the claim, the GHOST rule, and the fact that $|H|>n/2$, we have that $\vote(u_V)=B$.
Further, if $u_V'$ is the witness unit by $V$ in this round  (created at time $2R/3$) we also have $\vote(u_V')\geq B$, because\footnote{The reason why we claim that $\vote(u_V') \geq B$ and not the stronger version $\vote(u_V')=B$ is because a dishonest validator could have proposed another block on top of $B$ in the meantime.} $u_W\in L(u_V')$ for all $W\in H$ (because GST already happened, we know that each such $u_W$ is received by $V$ before $u_V'$ is created).
By the same reasoning, each witness unit $u''_V$ created by an honest validator in round $(r_0+1)$ (i.e., created at time $2R/3$) has all units $\{u'_W\}_{W\in H}$ in its downset and $\vote(u''_V)\geq B$.
By repeating this argument (see an analogous argument in the proof of Lemma~\ref{lemma:final_follow}) one concludes that for every $k>0$, by the start of round $r_0+k$ each honest validator reaches a state $\sigma'$ which contains a $(|H|, k)$-summit relative to $\sigma'$. 

Furthermore, by a calculation as in Lemma~\ref{lemma:final_follow}, we have that for $k>\log(t+1)$ such summits already finalize block $B$.
Consequently, $\final(B, \sigma', t)$ holds for each honest validator's state $\sigma'$ after round $r_0+k$.
Since this gives us a finalized block of height $\geq h+1$, we reach a contradiction with the assumption that no block of height $h+1$ is ever finalized.
\end{proof}

\subsection{Communication Complexity}
The so far described protocol, in the absence of Byzantine nodes is rather efficient when it comes to communication complexity.
Indeed, each validator produces and broadcasts $2$ units per round on average, and downloads the $\approx 2n$ units created by other validators in the same round.
The situation becomes much more tricky if some validators are Byzantine and try to make the communication complexity large, to possibly exhaust resources of honest validators, and cause crashes. 
One possibility for that would be to send a large number of invalid or malformed messages to honest validators -- while such attacks can be dangerous, they can be dealt with by simply discarding all such incorrect messages that are being received.
Another possibility is to send messages, i.e., units that are not produced according to the protocol rules, but still look correct and thus cannot be discarded by honest nodes.
As it turns out, the latter type of attack needs to be addressed already in the protocol, and not in the implementation, as otherwise the communication complexity of the protocol would be essentially infinite.
One can distinguish two basic {\it spam patterns} of this type that the dishonest validators can try:
\begin{itemize}
    \item {\bf Vertical Spam.} A dishonest validator does not equivocate but produces new units at a fast pace, violating the Unit Creation and Reception Strategy. This type of spam is easy to deal with, as one can compute exactly how many units a validator should have produced by a given moment in time. And if we are receiving more, then we can simply ignore them. 
    \item {\bf Horizontal Spam.} A dishonest validator or a group of dishonest validators produce a large number of equivocations and send them in parallel to different honest validators to make the DAG large. These would by definition be equivocations for which the attacker might be punished (likely by slashing its stake), but they could potentially be disruptive nonetheless.
\end{itemize}
The latter type of attack is by far more dangerous and harder to deal with.
In the most basic form of this attack, a single malicious validator sends a very large number of equivocating votes to the other validators.
If the other validators would include all of these votes in their local DAGs, this would soon exhaust all memory available on their machines.
We note that this kind of direct attack is fairly easy to avoid simply by having validators refuse to directly cite equivocating votes in their downsets (we will discuss a slightly more sophisticated version of this idea shortly).
This ensures that when there is {\bf only one attacker}, then each honest validator directly cites at most one chain of messages by the attacker.
This means that a single equivocating attacker can produce $\Omega(n)$ copies of its units (one per honest validator).
This is a reasonable bound, as such behavior is severely punished, and thus is expected to happen rarely if at all.

This however is not the end of the story, as the above defense does not quite work if {\bf multiple attackers} act together.
Suppose that $A_1,A_2,B_1,B_2,C_1,C_2,\ldots \in \cV$ are the attacking validators who work together in order to create the following collection of messages. 
\begin{enumerate}
    \item $A_1$ creates units $m_1$ and $A_2$ creates a unit $m_2$,
    \item $m_1$ directly cites units $m_{11}$ and $m_{12}$ by $B_1$ and $B_2$ respectively,
    \item $m_2$ directly cites equivocating units $m_{21}$ and $m_{22}$, again by $B_1$ and $B_2$,
    \item unit $m_{ij}$ is above $m_{ij1}$ and $m_{ij2}$ by $C_1$ and $C_2$.
\end{enumerate}
Continuing this pattern, a conspiracy of $2k$ validators can create a pattern of $2^k$ votes over the course of $k$ rounds so that no vote directly cites an equivocation.
An honest validator trying to validate the top units in this pattern will need to download exponentially many equivocating units at the bottom.
What is perhaps worse here is that although it will be easy for the honest validators receiving these units to tell that something is going wrong (given the extraordinary number of units and equivocations), it is non-trivial to determine exactly who is responsible and where to start cutting off the bad units.
For example, in the pattern described, $A_1$ and $A_2$ actually have not equivocated and so cannot be distinguished from honest validators.
This means that a validator $V$ receiving this pattern will be left with the choice of either ignoring these messages (which means that if $A_1$ and $A_2$ were honest that these validators are being permanently cut off from each other), or forwarding them (which can be a problem if a version of this with different equivocating messages was sent to the other honest validators).
The above attack, referred to as {\it Fork Bomb} has been described in~\cite{GLSS19}.

To deal with such an attack, we will need to be extra suspicious of units that cite other units that equivocate with ones that we already know of.
In particular, we will want to know that these votes are not themselves equivocations.
We resolve this by introducing a system of endorsements.

\subsubsection{Spam Prevention Using Endorsements}\label{sec:endorse}
We note that one equivocating validator $A\in \cV$ might send equivocating units to each of the other $n-1$ validators and if they do not coordinate about which units to include in their next downset, then all these may end up part of the DAG, and thus each honest validator will be forced to include $\Omega(n)$ chains of units from $A$ in order to incorporate each others' units.
Thus, as long as we do not want to coordinate between validators the inclusion of every single unit in the DAG, then we must accept that $\Omega(n)$ chains per validator (in case it equivocates) is the best we can hope for.

We propose a strategy that requires no coordination in the optimistic case -- when no equivocation is present in the DAG, and only adds some overhead in the pessimistic case -- when equivocations are there.
At the same time it achieves the best possible guarantee: no local DAG of an honest validator ever contains more than $\Omega(n)$ chains of units per equivocator.
To this end we introduce a new type of messages in the protocol: {\it Endorsements}.

An endorsement $\eendorse(V, u)$ is a message by a validator $V\in \cV$ (digitally signed by $V$) regarding a unit $u$.
Intuitively the meaning of this message is:
\begin{center}
    $\eendorse(V, u) ~~\equiv ~~$ {\it at the moment of creating this endorsement, $V$ was not aware of any equivocations by $S(u)$}
\end{center}
We say that a unit $u$ is endorsed if $>n/2$ validators $V$ sent endorsements for $u$.
Note that being endorsed is a subjective property as it depends on the local view of a validator.

\desc{

\begin{center}
    {\bf Endorsement strategy for $V\in \cV$ } 
\end{center}
\begin{enumerate}
\item $\cV$ starts the protocol in the {\it relaxed} mode.
\item After seeing any equivocation, $V$ switches to the {\it cautious} mode\footnote{There is no explicit exit condition for the cautious mode. However, as we specify in Section~\ref{sec:eras}, the protocol execution is split into eras, and at the start of each era, every validator is initialized in relaxed mode.}.
\item When entering {\it cautious mode}, for all units $u$ that $V$ is aware of, if $V$ is not aware of any equivocation by $S(u)$, $V$ broadcasts $\eendorse(V,u)$.
\item When being in {\it cautious mode}, whenever $V$ learns for the first time about a unit $u$ and is not aware of any equivocation by $S(u)$, $V$ broadcasts $\eendorse(V,u)$.
\item In any mode: whenever $>n/2$ endorsements $\eendorse(W,u)$ for a unit $u$, from pairwise different validators $W$ have arrived, set the status of $u$ as {\it endorsed}.
\end{enumerate}
}

We emphasize that the above strategy for sending endorsements is rather naive, as it sends endorsements for all ``honest'' units.
This version allows for a simple analysis, yet might be slightly inefficient in practice.
In Section~\ref{sec:practical} we show a refined strategy to achieve the same goal but sending less endorsements.

We note that an honest validator will never endorse a pair of equivocating units as by the time they receive the second of these units, they will know the sender to be an equivocator and thus will not endorse.
In fact, exhibiting endorsements of equivocating messages by a single validator is verifiable proof of bad behavior by that validator for which they can be penalized.
Furthermore, assuming as always that there are $f<\frac{n}{3}$ dishonest validators we can guarantee that  the maximum number of incomparable units by the same validator that are simultaneously endorsed is $3$. 
We will often think of endorsed units as being certified by a significant fraction of honest validators.
While the downsets of other units might be considered suspect (in that they might contain equivocation bombs), endorsed units are considered to be more solid.
\begin{definition}
A unit $u$ is said to cite another unit $v$ \emph{naively}, denoted $u >_n v$, if $u > v$ and there is no endorsed unit $w$ such that $u > w \geq v$.
\end{definition}

Using the notion of naive citation, we introduce a validity criterion for units that allows us to avoid equivocation bombs.

\desc{
\begin{center}
    {\bf Limited Naivety Criterion (LNC) } 
\end{center}
A unit $u$ by validator $V$ is considered incorrect if there is an equivocation $(v_1, v_2)$ (created by some validator $W$) and two units $u_1, u_2\leq u$ created by $V$ such that 
$u_1 >_n v_1$  and $u_2 >_n v_2.$
}

Since, as previously noted, endorsed units, and thus also naive citations are subjective notions, also satisfaction of LNC by a given unit might depend on the local view of a validator.
Importantly, the status of a unit for a given validator can change only one way - if a given unit satisfies LNC, it will never cease to satisfy it, if a unit is endorsed, it will never cease being endorsed, and if a unit is not cited naively by another unit $v$, it will never be considered to be cited naively by $v$.
Thus it is important to suitably modify the unit reception strategy so that units upon which a given validator cannot yet build its own unit without breaking LNC are kept in a buffer. 
Once such a unit gathers enough endorsements, it can be removed from the buffer and added to the DAG\footnote{To prevent unbounded growth of a buffer, we clear it at the end of each era as defined in Subsection \ref{sec:eras}}.

Perhaps the most important property of the Limited Naivety Criterion is that it allows to show that not too many equivocating votes are contained in any message downset.
More specifically, the following theorem states that we can bound the number of units below an honest unit.

\begin{theorem}\label{thm:limit_spam}
Assume that $f<\frac{n}{3}$. If $u$ is an $N$-th unit created by an honest validator, $u$ is above at most $O(n N(1 + f_{\mathrm{equiv}}))$ other units, where $f_{\mathrm{equiv}}=|E(D(u))|$.
\end{theorem}

\noindent 
The following subsection is devoted to proving Theorem~\ref{thm:limit_spam}.

\subsubsection{Bounding the Number of Equivocations}
We start with a simple lemma bounding the number of independent equivocations that can be endorsed at the same time.
\begin{lemma}\label{lem:limit_endors}
Suppose that $f<\frac{n}{3}$ and let $\sigma$ be the state of any honest validator during the protocol execution. %
Suppose that $C\subseteq \sigma $ is a set of pairwise incomparable units created by some validator $W$.
If all units in $C$ are endorsed, then $|C|\leq 3$. 
\end{lemma}
\begin{proof}
For each unit $u \in C$ denote by $E_u \subseteq \cV$ the set of honest validators who endorsed $u$.
We have $|E_u|+f >n/2$, as each such $u$ is endorsed by assumption.

Since honest validators never endorse two incomparable units from a single validator, it follows that the sets $E_u$ are disjoint, hence
$$n-f \geq \sum_{u\in C} |E_u| > |C|(n/2 - f),$$
and since $f<\frac{n}{3}$, it follows that $|C|<4$.
\end{proof}
The proof of Theorem~\ref{thm:limit_spam} relies on the following technical lemma that bounds the number of different equivocations that a dishonest validator can create, so that all of them are included in an honest validator's state.
\begin{lemma}\label{chainBoundLemma}
Assume that $f< \frac{n}{3}$.
Let $\sigma$ be the protocol state of an honest validator (thus all units in $\sigma$ satisfy the Limited Naivety Criterion). Then for any unit $u\in \sigma$, and for any validator $V$, the set of units created by $V$ in $D(u)$ is contained in the union of at most $O(n)$ chains of units.
\end{lemma}

\begin{proof}
Let $C$ be any set of pairwise incomparable units, created by single validator $V$ in the downset of $u$.
We will show that $|C|\leq 3f+(n-f)+1$. Once we have that, the lemma then follows from Dilworth's Theorem.

First of all, there might be at most one unit $u_0\in C$ such that some unit by $S(u)$ naively cites $u_0$ (this follows from $u$ satisfying LNC).
For clarity we remove such an $u_0$ from $C$ and later account for it by adding $+1$ to the obtained upper bound.
At this point, for each $v\in C$ there exists an endorsed unit $e$ such that $v\leq e \leq u$.
For each $v\in C$, we pick $e_v$ to be any minimal unit in the set of all endorsed units $e$ such that $v \leq e \leq u$.

We claim that:
\begin{itemize}
    \item The units $e_v$ for $v\in C$ are pairwise distinct.
    \item $|\{e_v : v\in C\}| \leq 3f+(n-f)$
\end{itemize}
It is easy to see that the lemma follows after establishing the claims.

For the first claim, if we had that $e_{v_1} = e_{v_2}$ for $v_1, v_2\in C$ with $v_1\neq v_2$, then $e_{v_1}$ would naively cite two equivocations (this follows from minimality of $e_{v_1}$), thus would violate LNC.

For the second claim, we first note that if $S(e_{v_1})=S(e_{v_2})$ for some $v_1, v_2 \in C$, then $e_{v_1}, e_{v_2}$ must be incomparable.
Indeed if we had, without loss of generality, that $e_{v_2} \leq e_{v_1}$ then $e_{v_1}$ would violate LNC, as $v_2 <_n e_{v_2}$ and $v_1<_n e_{v_1}$.
This means that, by Lemma~\ref{lem:limit_endors} each equivocator may contribute at most $3$ units to $\{e_v: v\in C\}$ and the remaining validators may contribute at most $1$.
In total, this gives $|\{e_v: v\in C\}|\leq 3f+(n-f)$.
\end{proof}

Lemma \ref{chainBoundLemma} provides a bound on the number of messages that can be created by means of equivocation spam. This is enough to prove Theorem~\ref{thm:limit_spam}.

\begin{proof}[Proof of Theorem~\ref{thm:limit_spam}]
Each chain of units created by a single validator in $D(u)$ might have length at most $O(N)$ since the creator of $u$ is honest and does not add units from the future to its DAG.
By Lemma~\ref{chainBoundLemma}), each equivocator in $D(u)$ might contribute at most $O(n)$ chains, $O(N)$ units each.
Thus there will be at most $O(nNf_{\mathrm{equiv}})$ units by equivocators in $D(u)$. 
The remaining honest validators produce at most $O(nN)$ units in total, during these $N$ rounds.
\end{proof}

\subsubsection{Liveness After Adding Spam Prevention Measures}\label{sec:dynamicrounds}

While the application of Limited Naivety Criterion ensures that an attacker cannot force the honest nodes to process too many units, it is no longer clear that it can be made to work with our original liveness strategy as that required validators to create units above all units they are aware of, which may well violate the Limited Naivety Criterion. We thus need to slightly modify our liveness strategy.

First of all we increase the length of a round from $R=3\Delta$ to $R=6\Delta$, the reason for that is that for a unit sent by an honest validator at time $0$, $R/3$ or $2R/3$ we want not only the unit but also its endorsements to reach each other honest validator by time $R/3$, $2R/3$ or $R$ respectively.
Further, whenever a validator is in the {\it suspicious} mode, we instruct it to directly cite only endorsed units (validator's own most recent unit is exempt from this rule).

For the proof of the main liveness result -- Theorem~\ref{thm:liveness} -- to go through, we need the following requirement to be satisfied: 

\noindent {\bf Rapid Endorsement Spread:} ``Whenever a unit $u$ is created by an honest validator $V$ after $\gst$, after $R/3$ time, each unit in $D(u)$ that was endorsed in $V$'s local view at the time of creating $u$ is also endorsed in the local view of every honest validator.''

To guarantee the above condition the simplest (but perhaps not the most efficient) way would be to instruct the honest validators to broadcast a set of $>n/2$ endorsements of a particular unit, after it becomes endorsed in their local view.
Various more efficient strategies here are possible, we discuss some more practical approaches to endorsements in Section~\ref{sec:practical}.

\section{Practical Considerations}\label{sec:practical}
\subsection{Dynamic Round Lengths}
When defining the base version of our protocol we divided the time into fixed length rounds.
The length of a round was picked so as to make sure that in each such round, even assuming worst case delays, a few communication ``round trips'' between honest validators would be possible.
On the theoretical side, the worst case message delay $\Delta$ was available from the partial synchrony assumption in the {\it known $\Delta$ flavour}, as initially defined in~\cite{DLS88}.
In practice, one can estimate an appropriate $\Delta$ using experiments, yet there are a few dangers to be aware of when setting a constant $\Delta$:
\begin{itemize}
    \item If we set the $\Delta$ too optimistic (i.e., too low) then liveness might be violated, as some messages between (high ping) pairs of validators ping can arrive too late.
    \item If we set the $\Delta$ too pessimistic, then the protocol might become quite slow, as during significant portion of the round, the validators would stay idle.
\end{itemize}

This motivates the use of dynamic round lengths so that the protocol by itself adapts to the current network conditions.
At a theoretical layer this ``approach'' corresponds to the {\it unknown $\Delta$ flavour} in the original definition of partial synchrony by~\cite{DLS88}.
The theoretical model states that there is some fixed, worst case message delay $\Delta$, yet it is not available at the start of the protocol execution.
Even though these two models of partial synchrony are equivalent via almost black-box reductions, we present here a custom version of our protocol for the unknown $\Delta$ model that is optimized for practical implementations.

\noindent {\bf Rules for changing round lengths.}
We measure time in milliseconds since the epoch\footnote{Unix time: number of seconds since the beginning of 1970, UTC, but millisecond resolution is more suitable for this purpose}.
An integer timestamp is called a \emph{tick}.
As before, without loss of generality, we assume that the participants have synchronized clocks.
We pseudorandomly assign one of the validators $\mathcal{L}(i) \in \mathcal{V}$ to each tick $i$ as the tick's \emph{leader}.
One needs to make sure that this schedule is sufficiently random, so that honest validators appear in this schedule frequently.

We assume that units have timestamps, i.e., for a unit $u$, $T(u)$ is the time when unit $u$ was sent.

Each validator $v$ maintains a private parameter $n_v(i) \in \mathbb{N}$ (called the {\it round exponent}) that is updated periodically. Formally, the map $n_v: \mathbb{N} \rightarrow \mathbb{N}$ assigns to each tick number the round exponent that the validator $v\in \cV$ should use at this particular tick.
Below we give a strategy on how to select $n_v$, but we always assume that $n_v (i) = n_v (i - 1)$ unless $i$ is a multiple of both $2^{n_v(i)}$ and $2^{n_v(i-1)}$. In other words, we assume that the parameters are kept constant for time windows of $2^{n_v(i)}$ ticks, from $j$ to $j + 2^{n_v(i)} - 1$, where $j \leq i$ is maximal such that it divides $2^{n_v(i)}$. 

For a validator $v\in \cV$ a new round starts whenever the current tick $i$ is divisible by $2^{n_v(i)}$ and this round's length is $R:=2^{n_v(i)}$.
Further, for every tick $i$, the value $i \mbox{ mod }2^{n_v(i)}$ determines the time that has passed since the round started.
For this round, this validator follows the unit reception and creation strategy as described in Section~\ref{sec:unit_creation}.
The leader for this round is $\cL(i_0)$ where $i_0$ is the tick at which the round started.

We emphasize that (as explained below) $n_v$, and thus also round lengths, are changed for each validator $v\in \cV$ separately, based solely on their local view.
For this reason, the rounds might (and surely will) fall out of sync from time to time, even between honest validators.
However, the round length adjustment strategy is designed in such a way so that in practice, this happens rarely and the round lengths will differ only by a factor of $2$.

\noindent {\bf Strategy for changing $n_v$.} Our strategy is parametrized by the following numbers:
\begin{itemize}
    \item two natural numbers $\nmin<\nmax$ that determine the minimum and the maximum, respectively, value of $n_v$ that we allow for validators,
    \item a confidence threshold $t_0$ that is used to measure progress in finalization,
    \item numbers $0<C_{\mathrm{fail}}< C_{\mathrm{succ}} <C$ and $D$, that specify the conditions based on which the round exponent changes.
\end{itemize}

\desc{

\begin{center}
    {\bf Round Exponent Maintenance Strategy for $V\in \cV$} 
\end{center}
\begin{enumerate}
\item At tick $i=0$, initialize $n_v(j)= \nmin$ for all $j$. Initialize $\cntsucc=0$.
\item At tick $i>0$,  let $m:=n_v(i)$. Whenever $i$ is divisible by $2^{m+1}$:
\begin{enumerate}
        \item Let $b_{\mathrm{fin}}$ denote the number of blocks finalized with confidence $t_0$ during the last $C$ rounds,
        \item if $b_{\mathrm{fin}}\leq C_{\mathrm{fail}}$ then set $n_v(j)=\min(m+1, \nmax)$, for all $j\geq i$,
        \item if $b_{\mathrm{fin}}\geq C_{\mathrm{succ}}$ then increment $\cntsucc=\cntsucc+1$, otherwise set $\cntsucc=0$.
        \item If $i$ is divisible by $C2^{m+1}$ and $\cntsucc\geq D$, set $n_v(j)=\max(m-1, \nmin)$, for all $j\geq i$. Set $\cntsucc=0$.
\end{enumerate}
\end{enumerate}
}

In practice $t_0$ should probably be much lower than the confidence threshold used for ``safety'', for instance $t_0\approx 0.01$.
Even if in a single round, the block proposed by the leader gets finalized with the lower confidence $t_0$, if there are enough honest validators this confidence will increase in the future, as the summit grows in height. The constants $C_{\mathrm{fail}}, C_{\mathrm{succ}}, C, D $ are chosen so as to make sure that the ``correct'' round exponent is learned by the validators based on the ``finality rate'' feedback.
Example values of these constants could be $(C_{\mathrm{fail}}, C_{\mathrm{succ}},C) = (10, 32, 40)$ and $D=3$.

\subsection{Eras}\label{sec:eras}
Since the Highway protocol is meant for creating and maintaining a blockchain, once run (initialized) it is supposed to run forever, without stops.
Consequently, the validators are forced to store the whole DAG, even the units that were created at the very beginning of the protocol execution.
Erasing ``old'' parts of the DAG is not safe, as (likely dishonest) validators might directly cite in their recent units some units that are very old.

When talking about the storage issue, it is important to emphasize that clearly, the blockchain must be stored anyway, and this is a cost that is impossible to avoid.
However, intuitively, nothing should stop us from erasing old ``metadata'', i.e., units that have been used to finalize certain blocks, and thus are not useful anymore.
This issue becomes all the more serious after realizing that for large committees of validators, this ``metadata'' might likely require more space than the blockchain itself.
What is even worse is that such a DAG must be kept in RAM and not on disk.
The reason for that is that one must build efficient data structures over the DAG for the sake of detecting finality and checking correctness of units (in particular the LNC).
Placing such data structures on disk instead, would cause a dramatic slowdown. 

A practical solution for this issue is to divide the protocol exucution into {\it eras}.
In each era, $K=1000$ new blocks are added to the blockchain, and importantly a new instance of Highway is run in every era.
This means that, for instance, in era $5$, we consider the block $B_{4999}$ at height $4999$, finalized in era $4$, as the genesis unit.

If we use eras in such a way, then the validators only need to store blocks that were finalized in the previous eras and the DAG for the current era.
This also helps a lot with equivocation spam attacks, as equivocators caught in a particular era can be banned from the very beginning of the subsequent eras.
The honest nodes start a new era in the {\it relaxed} mode and thus do not need to send endorsements because of past equivocations anymore.

One additional benefit of using eras is that it allows to change the validator set according to some prespecified rules, and hence move the protocol towards the permissionless model. 

\subsection{Sending Less Endorsements}
Recall that the original endorsement strategy that we introduced in Section~\ref{sec:endorse} was very simple: after seeing an equivocation, endorse every unit by non-equivocating validators.
While it allows for quite simple arguments that liveness is preserved, it also introduces a non-negligible overhead in case when equivocations are detected in a particular era.

Below we present a refined strategy that still guarantees that no honest validator ever gets stuck because of lack of endorsements of its own, honest units.
At the same time, this strategy allows to send much fewer endorsements.

\desc{

\begin{center}
    {\bf Refined Endorsement Strategy for $V\in \cV$ } 
\end{center}
\begin{enumerate}
\item Initialize the set $\equivs:=\emptyset$.
\item Whenever $V$ sees an equivocation by $W$, add $W$ to $\equivs:=\emptyset$.
\item For every unit $u$ that $V$ is aware of, such that the following conditions are satisfied, $V$ broadcasts   $\eendorse(V,u)$:
\begin{enumerate}
    \item The creator of $u$, $S(u)$ does not belong to $\equivs$.
    \item There exists a validator $W\in \equivs$ such that
    \begin{itemize}
        \item $W\notin E(u)$,
        \item There is a unit $w$, created by $W$ such that $w\in D(u)$ but the unit $u'$ created by $S(u)$, just before $u$, does not cite $w$.
    \end{itemize}
\end{enumerate}
\end{enumerate}
}

In a practical scenario, after an equivocation by a validator $W$ is detected at round $r$, then after a few rounds no honest validator will cite any new units by $W$ anymore.
The above strategy requires, in such a case, sending endorsements only for units created in these few rounds that happen after round $r$.
The honest units that are created later, never cite any ``new'' units by $W$ and thus do not require endorsements.

\subsection{Weighted Consensus}

So far it has been assumed that the opinion of every validator is equally important in the process of achieving consensus.
In this subsection, we describe the modifications allowing the Highway protocol to be run in a scenario where each validator has an associated integer weight, corresponding to its voting power -- a useful version for example when constructing a Proof of Stake blockchain. 

Let the sum of all validator weights equals $N$, and the function $\ww(V)$ denotes the weight associated with validator $V$.
In such a version, the function $\mathrm{total}$ being the basis for the GHOST rule has to be replaced with the one counting the \emph{total weight} of validators supporting a given option, instead of their number. 
Similarly, the \textbf{density} condition in the definition of the summit needs to be replaced with the version bounding total weight of the set of validators  $S\big(\bar{D}(u) \cap C'_i)\big)$, instead of their number. 
After these changes are introduced, Theorems \ref{thm:finality} and \ref{thm:liveness} remain true without any changes, as they do not deal with number of validators, and the proofs do not take advantage of the equal validator weight. 

Note that while the finality in such a modified scenario is purely a function of validator weights, the communication complexity remain dependent on the total \emph{number} of validators -- every validator needs to download the units of every other validator, no matter how small its weight is.
For this reason, when desigining the blockchain system, the number of validators should be considered as the main factor influencing latency and throughput, not their weights and dependencies between them.

\bibliographystyle{alpha}
\bibliography{highway}

\appendix

\end{document}